\documentclass[12pt]{article}

\usepackage{amsmath,amssymb,amstext,amsthm}

\newcommand{\p}[1]{\mathbf{Pr}\left[{#1}\right]}
\newcommand{\e}[1]{\mathbf{E}\left[{#1}\right]}
\renewcommand{\v}[1]{\mathbf{Var}\left[{#1}\right]}

\theoremstyle{plain}
\newtheorem{theorem}{Theorem}
\newtheorem{lemma}[theorem]{Lemma}
\newtheorem{proposition}[theorem]{Proposition}

\theoremstyle{definition}

\theoremstyle{remark}

\title{On the Stretch Factor of \\ Randomly Embedded Random Graphs}
\date{}
\author{Abbas Mehrabian\thanks{email: \texttt{amehrabi@uwaterloo.ca}}\and Nick Wormald\thanks{Supported by the Canada Research Chairs Program and NSERC. email: \texttt{nwormald@uwaterloo.ca}} \\
{\small {Department of Combinatorics and Optimization, University of Waterloo} }}

\begin{document}

\maketitle

\begin{abstract} 
We consider a random graph ${\cal G}(n,p)$ whose vertex set $V$ has been randomly embedded in the unit square and whose edges are given weight equal to the geometric distance between their end vertices. Then each pair $\{u,v\}$ of vertices have a distance in the weighted graph, and a Euclidean distance. The stretch factor of the embedded graph is defined as the maximum ratio of these two distances, over all $\{u,v\}\subseteq V$. We give upper and lower bounds on the stretch factor (holding asymptotically almost surely), and show that for $p$ not too close to 0 or 1, these bounds are best possible in a certain sense.
Our results imply that the stretch factor is bounded with probability tending to 1 if and only if $n(1-p)$ tends to 0,
answering a question of O'Rourke.
\end{abstract}

\section{Introduction}

Let $G$ be a graph embedded in the plane.
For every two points $u$ and $v$, let $d(u,v)$ denote their Euclidean distance.
Make $G$ weighted by putting weight $d(u,v)$ on every edge $uv$.
For two vertices $u$ and $v$, let $d_G(u,v)$ denote
their shortest-path distance on (weighted) $G$.
The \emph{stretch factor} of $G$ is defined as
$$\max \: \frac{d_G(u,v)}{d(u,v)}\:,$$
where the maximum is taken over all vertices $u,v$.
If $G$ is disconnected then its stretch factor is undefined.

The stretch factor (also known as the \emph{spanning ratio} or the \emph{dilation}) is a
well studied parameter in discrete geometry, see for instance the book~\cite{spanners_book}
or the recent survey~\cite{spanners_survey}.
An important problem in this context is the following.
Given $n$ points on the plane, find a set of $O(n)$ pairs of them, 
such that when you create a geometric graph by adding
the segments joining the points in each pair, 
this geometric graph has bounded stretch factor.
A possible approach is to choose a random set of pairs.
Suppose that we randomly choose $M$ distinct pairs from the set of all $\binom{n}{2}$ pairs of points,
and add the corresponding edges.
Then, one can ask, how large should $M$ be to guarantee that the stretch factor is bounded, with probability tending to 1?
In this paper we show that if the initial points are chosen uniformly at random from the unit square,
then we need almost all edges to guarantee a bounded stretch factor,
hence this method is inefficient.

The setting is as follows.
Select $n$ points uniformly at random from the unit square,
and then form a random geometric graph $G$ on these points by joining each pair
independently with probability $p$, where $p$ is in general a function of $n$.
This is not a ``random geometric graph'' in the sense of Penrose~\cite{penrose},
because points are joined without regard to their geometric distance.
However, one can call this a \emph{randomly embedded random graph},
since you get the same thing if you start from an Erd\H{o}s-R\'{e}nyi random graph with parameters $n,p$
and embed each of its vertices into a random point in the unit square.
The stretch factor of $G$ is a random variable and we denote it by $\mathcal{F}(n,p)$.
We study the asymptotic behaviour of $\mathcal{F}(n,p)$ when $n$ is large,
and give probabilistic lower and upper bounds for it.
In the following, \emph{asymptotically almost surely} means with probability $1-o(1)$,
where the asymptotics is with respect to $n$.

In the open problem session of CCCG 2009~\cite{open},
O'Rourke asked the following question:
for what range of $p$ is $\mathcal{F}(n,p)$ bounded asymptotically almost surely?
As a conclusion of our bounds,
we answer this question as follows.
Let $\lambda > 1$ be any fixed constant, and note that $\omega(1)$ denotes a function that tends to infinity as $n$ grows.
If $n(1-p) = \omega(1)$, then
asymptotically almost surely $\mathcal{F}(n,p) > \lambda$.
If $n(1-p) = \Theta(1)$, then
$\mathcal{F}(n,p) > \lambda$ with probability $\Omega(1)$.
Finally, if $n(1-p) = o(1)$, then
asymptotically almost surely $\mathcal{F}(n,p) < \lambda$.

Our main lower bound is the following theorem.

\begin{theorem}
\label{thm:lower_bound}
Let $w(n) = \omega(1)$.
Then asymptotically almost surely
$${\mathcal{F}(n,p)} > \frac{\sqrt{n(1-p)}}{w(n)}\:.$$
\end{theorem}

Let $\lambda$ be fixed.
This theorem implies that if $n(1-p) = \omega(1)$, then
asymptotically almost surely $\mathcal{F}(n,p) > \lambda$.
This strengthens the result of the first author~\cite{conference},
who proved the same thing for $p < 1 - \Omega(1)$.

Let $CON$ denote the event {``$G$ is connected.''}
Recall that if the graph is disconnected, then its stretch factor is undefined.
For any $p<1$, this happens with a positive probability, hence $\mathbf{E}[\mathcal{F}(n,p)]$
is undefined.
It is then natural to bound $\mathbf{E} [ \mathcal{F}(n,p) | CON ]$ instead.

Our main upper bound is the following theorem.

\begin{theorem}
\label{thm:upper_bound}
Let $p^2n \geq 33 \log n$ and let $w(n)=\omega(1)$. Then, asymptotically almost surely we have
$$\mathcal{F}(n,p) \leq 1 + \frac {w(n) \sqrt{n(1-p)}}{p} \:.$$
If $p^2n \geq 113 \log n$, then
$$\e{\mathcal{F}(n,p)|\:CON} \leq 1 + {\frac{\sqrt{2048 n (1-p) }}{p}} + o(1)\:.$$
\end{theorem}

When $n(1-p) = o(1)$, this theorem implies that for any fixed $\epsilon>0$,
we have that $\e{\mathcal{F}(n,p) | CON} \leq 1+\epsilon$, and asymptotically almost surely $\mathcal{F}(n,p) \leq 1 + \epsilon$.

The more interesting case is when $n (1-p) = \Theta(1)$.
In this regime Theorem~\ref{thm:upper_bound} states that $\e{\mathcal{F}(n,p)|\:CON} = O(1)$.
So, one may wonder if there is a constant $\lambda$ such that asymptotically almost surely $\mathcal{F}(n,p) < \lambda$.
However, Lemma~\ref{lem:lower_bound} (which is the main lemma in the proof of Theorem~\ref{thm:lower_bound})
implies that this is not the case:
for any fixed $\lambda$, with probability $\Omega(1)$ we have $CON$ and $\mathcal{F}(n,p) > \lambda$.
In other words, the random variable $\mathcal{F}(n,p)$ is not concentrated.
In this case one might expect that the distribution of $\mathcal{F}(n,p)$ tends to some nontrivial limit if $n(1-p)$ is constant.

Lemma~\ref{lem:lower_bound} actually implies that for a wide range of $p$, 
the first conclusion of Theorem~\ref{thm:upper_bound} is tight, in the sense that $w(n)$ cannot be replaced with a constant.
Namely, the following is true.

\begin{theorem}
\label{thm:tightness}
Assume that $p=\Omega(1)$ and $n(1-p)=\Omega(1)$. There is no absolute constant $C$ for which asymptotically almost surely
$$\mathcal{F}(n,p) \leq \frac {C \sqrt{n(1-p)}}{p} +1\:.$$
\end{theorem}

There is a nontrivial gap between our lower and upper bounds when $p=o(1)$.
It remains open to determine which of the bounds are closer to the correct answer in this regime.

The following notation will be used in the rest of the paper.
For a point $Q$ and nonnegative real $R$,
$C(Q, R)$ denotes the intersection of the disc with centre $Q$ and radius $R$ and the unit square, and
$F$ simply denotes $\mathcal{F}(n,p)$.
We often identify each vertex with the point it has been embedded into.
All logarithms are in natural base.

\section{The Lower Bound}
In this section we prove Theorems~\ref{thm:lower_bound}~and~\ref{thm:tightness}.
First we need an easy geometric result.

\begin{proposition}
\label{prop:intersection}
Let $Q$ be a point in the unit square.
If $0 \leq R \leq 1/2$,
then $C(Q,R)$ has area at least $\pi R ^2 / 4$.
If $0 \leq R \leq \sqrt 2$,
then $C(Q,R)$ has area at least $\pi R ^2 / 32$.
\end{proposition}

\begin{proof}
By symmetry, we may assume that $Q$ lies in the upper left quarter of the unit square.
If $0 \leq R \leq 1/2$, then the bottom right quarter of the disc with centre $Q$ and radius $R$
lies completely inside the unit square, and hence the intersection area is at least $\pi R^2 / 4 \geq \pi R^2 / 32$.
If $1/2 < R \leq \sqrt 2$, then $C(Q,R)$ contains $C(Q,1/2)$, so its area is at least
$$\frac{\pi \left(1/2\right)^2}{4} = \frac{\pi}{16} = \frac{\pi (\sqrt 2)^2}{32} \geq \frac{\pi R^2}{32}\:.\qedhere$$
\end{proof}

Let $c$ be such that $1/ 51 < c < 1/ 16\pi$ and $cn$ is an even integer.
Notice that here, as in the rest of the paper, we always mean $1/(ab)$ when we write $1/ab$.

\begin{lemma}
\label{lem:isolates}
Choose $cn$ points independently and uniformly at random from the unit square.
Build a graph $H$ on these vertices, by joining two vertices if their distance is at most $2 / \sqrt n$.
With probability at least $1 - O(1/n)$, $H$ has at least $cn / 2$ isolated vertices.
\end{lemma}

\begin{proof}
Let $X$ be the number of edges of $H$.
Then we have $X = \sum_{i<j} X_{i,j}$, where $X_{i,j}$ is the indicator variable for the distance between
vertices $i$ and $j$ being at most $2 / \sqrt n$.
Note that if vertex $i$ has been embedded in point $p_i$,
then $X_{i,j} = 1$ if and only if vertex $j$ is embedded in $C(p_i, 2/\sqrt n)$.
So by Proposition~\ref{prop:intersection} we have
$$\pi / n \leq \e{X_{i,j}} \leq 4 \pi / n\:.$$
Let $q = \e{X_{i,j}}$, $q_2 = 4 \pi / n$, and $M = \binom{cn}{2}$.
Thus
$$\e{X} = q M  =\Theta(n) \:.$$

We claim that $\v{X} = O(n)$.
Note that if $\{i,j\}$ and $\{k,l\}$ are two disjoint sets of vertices, then
$$\e{X_{i,j}X_{k,l}} = \e{X_{i,j}}\e{X_{k,l}} = q^2\:,$$
and the number of such pairs of pairs equals $\binom{cn}{2}\binom{cn-2}{2} \leq M^2$.
Otherwise, let $j = l$.
Then for $X_{i,j}=X_{k,j}=1$ to happen, both vertices $i$ and $k$ should be embedded at distance
at most $2 / \sqrt n$ from where vertex $j$ has been embedded, the probability of which is not more than $q_2 ^ 2$.
Hence in this case
$$\e{X_{i,j}X_{k,l}} \leq q_2 ^ 2\:,$$
and the number of such pairs of pairs is not more than $(cn)^3$.
Consequently,
$$\e{X(X-1)} = \sum_{\{i,j\} \neq \{k,l\}} \e{X_{i,j} X_{k,l}} \leq M^2 q^2 + (cn)^3 q_2^2\:,$$
and so
$$\v{X} = \e{X(X-1)} + \e{X} - \e{X}^2 \leq M^2 q^2 + (cn)^3 q_2^2 + q M - q^2 M^2 = O(n) \:.$$

By Chebyshev's inequality,
$$\p { X > 2 \e{X} } \leq \frac{ \v{X}}{\e{X}^2} = O(1/n)\:.$$
Thus with probability $1 - O(1/n)$, $H$ has at most $2 \e{X} \leq 4 c^2 \pi n$ edges.
If this is the case, then it has at least
$$ cn - 8 c^2 \pi n \geq cn / 2$$
isolated vertices, and this completes the proof.
\end{proof}

Now, back to the main problem.
Assume that the $n$ vertices are embedded one by one, and the edges are exposed at the end.
Consider the moment when exactly $cn$ vertices have been embedded.
Build an auxiliary graph on these vertices, by joining two vertices if their distance is at most $2 / \sqrt n$.
By Lemma~\ref{lem:isolates}, with probability $1-O(1/n)$, this graph has at least $cn/2$ isolated vertices.
We condition on the embedding of the first $cn$ vertices such that this event holds.
Let $A$ be a set of $cn/2$ isolated vertices in this graph.
The vertices in $A$ are called the \emph{primary} vertices,
the vertices that are one of the first $cn$ vertices but are not in $A$ are called \emph{far} vertices,
and the vertices that have not been embedded yet are called the \emph{secondary} vertices.

Let $m = cn/2$ be the number of primary vertices, and $n' = n(1-c)$ be the number of secondary vertices.
A \emph{primary disc} is a set of the form $C(v, 1/\sqrt n)$, where $v\in A$;
the vertex $v$ is called the \emph{centre} of the primary disc.
So, we have $m$ primary discs, say $\mathcal{R}_1, \mathcal{R}_2, \dots, \mathcal{R}_m$.
Let $\mathcal{W}$ be the set of points of the unit square that are not contained in any primary disc.
Notice that by the definition of $A$, the primary discs are disjoint, and no far vertex is contained in any primary disc.

Consider the following process for embedding the secondary vertices and exposing all of the edges of $G$:
\begin{enumerate}
\item
Note that $\{\mathcal{W},\mathcal{R}_1,\mathcal{R}_2,\dots,\mathcal{R}_m\}$ is a partition of the unit square.
In the first phase, to each secondary vertex, independently, we randomly assign an element of the partition,
with probability proportional to the area of that element.
Thus for each primary disc it is known how many secondary vertices it contains,
but their exact position is not known.
\item
In the second phase, for each secondary vertex we choose a random point in the corresponding element,
and place the vertex at that point.
\item
In the third phase, for every pair of vertices we add an edge independently with probability $p$.
\end{enumerate}
Clearly this process generates random geometric graphs with the same distribution as before,
however it makes the analysis easier.

\begin{lemma}
\label{lem:twopointregions}
With probability $1-\exp \left( -\Omega(n) \right)$, after the first phase, there exist at least $e^{-8} m$ primary discs
containing exactly two vertices: one primary (the centre) and one secondary.
\end{lemma}

To prove this lemma we will use the following large deviation inequality,
which is Corollary~2.27 in Janson~et~al.~\cite{JLR}.

\begin{proposition}[\cite{JLR}]
\label{thm:mcdiarmid}
Let $Z_1,Z_2,\dots,Z_{n}$ be a sequence of independent random variables,
and suppose that the function $f$ satisfies
$$|f(x_1,x_2,\dots,x_{n}) - f(y_1,y_2,\dots,y_{n})| \leq c\:,$$
whenever the vectors $(x_1,x_2,\dots,x_{n})$ and $(y_1,y_2,\dots,y_{n})$
differ only in one of the coordinates.
Then,
$$\mathbf{Pr} [ f(Z_1,Z_2,\dots,Z_{n}) - \mathbf{E}[f(Z_1,Z_2,\dots,Z_{n})] < -t ] < \exp\left(\frac{-t^2}{2nc^2}\right)\:.$$
\end{proposition}

\begin{proof}[Proof of Lemma~\ref{lem:twopointregions}.]
Recall that we have $n'$ secondary vertices,
and to each of them, independently, we randomly assign an element of the partition
$\{\mathcal{W},\mathcal{R}_1,\mathcal{R}_2,\dots,\mathcal{R}_m\}$,
with  probability proportional to the area of that element.
By Proposition~\ref{prop:intersection},
the area of each primary disc is between $\pi / 4n$ and $\pi / n$;
so for every $1\leq i \leq m$, the probability that $\mathcal{R}_i$ contains exactly one secondary vertex is at least
$$\binom{n'}{1} \: \frac{\pi}{4n} \: \left ( 1 - \frac{\pi}{n}\right)^{n'-1}
\geq \frac{\pi n'}{4n} \: \left ( \exp(-{2\pi}/{n})\right)^{n'-1}
\geq \frac{\pi n'}{4n} \exp \left ( -2 n' \pi /  n \right) \geq \frac{\pi}{8}\exp(- 2\pi)\:,$$
as $n' \geq n/2$. Let $p_1 = \pi\exp(-2 \pi ) / 8$.

Let $X$ be the number of primary discs that contain exactly one secondary vertex.
Since every primary disc contains exactly one primary vertex (its centre),
to prove the lemma we need to show that we have $X \geq e^{-8}m$ with probability $1-\exp \left( -\Omega(n) \right)$.
The calculation above shows that $\e{X} \geq m p_1$.
Since $p_1  > 2 e^{-8}$, showing that
$$\p{X < \e{X} / 2} \leq \exp \left( -\Omega(n) \right)$$
completes the proof of the lemma.

To every secondary vertex $v$ we assign a variable $Z_v$,
which equals $k$ if $v$ is assigned to $\mathcal{R}_k$ in the first phase,
and equals $0$ if $v$ is assigned to $\mathcal{W}$.
Since the assignment in the first phase is done independently,
the random variables $\{Z_v : v\mathrm{\ secondary}\}$ are independent.
Moreover, changing the value of $Z_v$ for a single vertex
will change $X$ by at most 2.
Hence by Proposition~\ref{thm:mcdiarmid} we have
\begin{align*}
\p {X < \e{X}/2} & \leq \exp \left( \frac{- \e{X}^2 / 4}{8n'} \right)  \leq \exp \left( \frac{- m^2 p_1^2 }{32n'} \right) = \exp \left( -\Omega(n) \right)\:.
\end{align*}
Therefore, with probability $1-\exp \left( -\Omega(n) \right)$, there exist at least $e^{-8} m$ primary discs
containing exactly two vertices.
\end{proof}

A primary disc $\mathcal{R}$ containing exactly two vertices $u$ and $v$ is called \emph{nice} if
the distance between $u$ and $v$ is less than $1 / \lambda \sqrt n$,
and $u$ and $v$ are not adjacent in $G$.
We claim that the existence of a nice primary disc $\mathcal{R}$ implies that the stretch factor of $G$ is larger than $\lambda$.
To see this, assume, by symmetry, that $u$ is the vertex of the centre of $\mathcal{R}$.
The (weighted) distance between $u$ and $v$ in $G$ is at least $1 / \sqrt n$,
since any $(u,v)$-path in $G$ must go out of $\mathcal{R}$ at the very first step.
However, the Euclidean distance between $u$ and $v$ is at most $1 / \lambda \sqrt n$,
and we have
$$\frac{d_G(u,v)}{d(u,v)}  > \lambda\:.$$

Theorem~\ref{thm:lower_bound} follows immediately from the following lemma.

\begin{lemma}
\label{lem:lower_bound}
For any positive $\lambda$ we have
$$\p{\mathcal{F}(n,p) < \lambda } \leq \exp \left[ - \frac{c n (1-p)}{2e^8 \lambda^2}\right] + o(1) \:.$$
\end{lemma}

\begin{proof}
Consider a primary disc $\mathcal{R}$ such that after the first phase,
it has been determined that $\mathcal{R}$ contains exactly two vertices, $u$ and $v$, where $u$ is the centre of $\mathcal{R}$.
Then for $\mathcal{R}$ to be nice, in the second phase
$v$ should be placed in $C(u, 1 / \lambda \sqrt n)$,
and in the third phase $u$ and $v$ should become nonadjacent in $G$.
The probability of the former is at least
$$\left(\frac{(1 / \lambda \sqrt n)^2}{(1/\sqrt n)^2} \right) = 1 / \lambda^2\:,$$
even if the disc of radius $1/\sqrt n$ centred at $u$ is not wholly contained in the unit square,
and the probability of the latter is $1-p$.
These two events are independent, so the probability that $\mathcal{R}$ is not nice is
at most $1-(1-p)  \lambda ^{-2}$.

By Lemma~\ref{lem:twopointregions}, once the first phase finishes, with probability $1 - \exp(-\Omega(n))$
there exists a set $\mathcal{B}$ of at least $e^{-8}m$ primary discs,
such that each primary disc in $\mathcal{B}$ contains exactly two vertices.
We condition on this event in the following.
The crucial observation is that the events happening during the second and third phases
inside each primary disc in $\mathcal{B}$ are independent of the others.
In particular, the events
$$\{\mathcal{R}\mathrm{\ is\ nice} : \mathcal{R}\in \mathcal{B}\}$$
are mutually independent;
hence the probability that none of the primary discs in $\mathcal{B}$ become nice during the second and third phases, is at most
$$\left [ 1-(1-p)  \lambda ^{-2} \right]^{e^{-8} m} \leq \exp \left[ - \frac{c n (1-p)}{2e^8 \lambda^2}\right]\:,$$
so that
$$\p{\mathcal{F}(n,p) < \lambda } \leq \exp \left[ - \frac{c n (1-p)}{2e^8 \lambda^2}\right] + O(1/n) + \exp(-\Omega(n))\:.\qedhere$$
\end{proof}

Theorem~\ref{thm:tightness} follows from Lemma~\ref{lem:lower_bound} by
putting $\lambda = C' \sqrt{n(1-p)}$ for a suitable constant $C'$,
noting that $p = \Omega(1)$ and $n(1-p) = \Omega(1)$.

\section{The Upper Bound}
In this section we prove Theorem~\ref{thm:upper_bound}.
We will use the following version of the Chernoff bound.
This is Theorem 2.1 in Janson~et~al.~\cite{JLR}.

\begin{proposition}[\cite{JLR}]
\label{prop_chernoff}
Let $X = X_1 + \dots + X_m$,
where the $X_i$ are independent identically distributed indicator random variables.
Then for any $ \epsilon \geq 0$,
$$\p{X \leq (1-\epsilon) \e{X}} \leq \exp (-\epsilon^2 \e{X} / 2 )\:.$$
\end{proposition}

\begin{lemma}
\label{lem:integration}
For any positive $\lambda$,
$$\p{F > 2\lambda + 1} \leq n^2 \left [\exp\left(- \frac{p^2n}{16}\right) + \frac{128 (1-p) }{p^2 n \lambda^2} \right]\:.$$
\end{lemma}

\begin{proof}
Say a pair $(u,v)$ of vertices is \emph{bad} if $d_G(u,v) > (2\lambda+1) d(u,v)$.
Let $u$ and $v$ be arbitrary vertices.
First, we show that with probability at least $1 - \exp(-p^2n / 16)$,
$u$ and $v$ have at least $p^2 n / 4$ common neighbours.
The expected number of common neighbours of $u$ and $v$ is $p^2 (n-2) > p^2 n /2$,
and since the edges appear independently,
by Proposition~\ref{prop_chernoff}, the probability that $u$ and $v$ have less than $p^2 n / 4$ common neighbours is less than
$\exp(-p^2 n / 16)$.
In the following, we condition on the event that $u$ and $v$ have at least $p^2 n /4$ common neighbours.

Now, consider the random embedding of the graph.
For any $t\geq 0$, if $u$ and $v$ are adjacent, or if
$d(u,v) \geq t$ and $u$ and $v$ have a common neighbour $w$ with $d(u,w) \leq \lambda t$,
then we would have $d_G(u,v) \leq (2\lambda+1) d(u,v)$ so the pair is not bad.
To give an upper bound for the probability of badness of the pair,
we compute the probability that $u$ and $v$ are nonadjacent, and
$i h \leq d(u,v) \leq (i+1) h$, and they have no common neighbour $w$ with $d(u,w) \leq \lambda ih$,
and sum over $i$.

Let us condition on the embedding of vertex $u$, and
denote by $a(u,s)$ the area of the set of points in the unit square at distance at most $s$ from $u$, and let $q=1-p$.
Then since $u$ and $v$ have at least $p^2n/4$ common neighbours and these common neighbours are embedded independently, for any $h>0$
$$\p{(u,v)\ \mathrm{bad}} \leq  \sum_{i=0}^{\lfloor 2 / h \rfloor } q \left[a(u,(i+1)h ) - a(u,ih)\right] (1-a(u,\lambda i h))^{p^2n / 4}\:.$$
Note that $\lim_{h\rightarrow 0}\frac { a(u,(i+1)h) - a(u,ih)}{h} \leq 2\pi ih $ and also for $ih > \sqrt 2 / \lambda$ the summand is zero. 
Hence, letting $t = ih$ and taking the limit,
$$\p{(u,v)\ \mathrm{bad}} \leq \int_{t=0}^{\sqrt 2 / \lambda} 2 \pi q t  (1-a(u,\lambda t))^{p^2n / 4}\:\mathrm{d}t\:.$$
By Proposition~\ref{prop:intersection}, $a(u,\lambda t) \geq \pi (\lambda t)^2/32$. Hence
\begin{align*}
\p{(u,v)\ \mathrm{bad}} & \leq \int_{t=0}^{\sqrt 2 / \lambda} 2 \pi q t  (1- \pi (\lambda t)^2/32)^{p^2n / 4}\:\mathrm{d}t =
\left. -\frac{128q}{(p^2n + 4)\lambda^2}\left(1-\frac{\pi \lambda^2 t^2}{32}\right)^{1 + \frac{p^2n}{4}} \right|_{t=0}^{\sqrt{2}/\lambda}\\
& = \frac{128q}{(p^2n + 4)\lambda^2}\left [ 1 - \left(1-\frac{\pi}{8}\right)^{1 + \frac{p^2n}{4}} \right]
< \frac{128q}{p^2n \lambda^2}\:.
\end{align*}
So by the union bound
$$\p{F>2\lambda+1} = \p{\exists\ \mathrm{a\ bad\ pair}} \leq n^2 \left [\exp\left(- \frac{p^2n}{16}\right) + \frac{128 (1-p) }{p^2 n \lambda^2} \right]\:.\qedhere$$
\end{proof}

We are now ready to prove Theorem~\ref{thm:upper_bound}.

\begin{proof}[Proof of Theorem~\ref{thm:upper_bound}]
We need to show that
if $p^2n \geq 33 \log n$ and $w(n)=\omega(1)$, then  asymptotically almost surely we have
$$F \leq \frac {w(n) \sqrt{n(1-p)}}{p} +1\:,$$
and that if $p^2n \geq 113 \log n$, then
$$\e{F|\:CON} \leq 1 + {\frac{\sqrt{2048 n(1-p)}}{p}} + o(1)\:.$$

For the first part,
let $\lambda = \frac {w(n) \sqrt{n(1-p)}}{2p}$.
By Lemma~\ref{lem:integration},
$$\p{F > \frac {w(n) \sqrt{n(1-p)}}{p} +1} = \p{F > 2\lambda + 1}
\leq n^2 \exp\left(- \frac{p^2n}{16}\right) + \frac{128 n(1-p) }{p^2 \lambda^2} = o(1) \:.$$
Thus, asymptotically almost surely $F \leq \frac {w(n) \sqrt{n(1-p)}}{p} +1$.

For the second part,
let $\epsilon = \sqrt{\frac{512 n(1-p) }{p^2}}$.
Since $F$ is nonnegative, we have
\begin{align}
\mathbf{E} [F|\: CON] & = \int_{s =0}^{\infty} \mathbf{Pr} [F > s |\: CON] \:\mathrm{d} s \nonumber \\
& = \int_{s =0}^{1+\epsilon} \mathbf{Pr} [F > s |\: CON] \:\mathrm{d} s
+ \int_{s =1+\epsilon}^{n^5} \mathbf{Pr} [F > s |\: CON] \:\mathrm{d} s
+ \int_{s =n^5}^{\infty} \mathbf{Pr} [F > s |\: CON] \:\mathrm{d} s.
\label{eq:EF2}
\end{align}
Clearly,
$$\int_{s =0}^{1+\epsilon} \mathbf{Pr} [F > s |\: CON] \:\mathrm{d} s \leq 1 + \epsilon\:.$$

Assuming the graph is connected, for any pair $\{u,v\}$ of vertices  there is a path having at most $n-1$ edges joining them,
so $d_{G}(u,v) < n \sqrt 2$.
Hence, if $F > s$, then there exists a pair $\{u,v\}$ of nonadjacent vertices with $d(u,v) < n\sqrt 2 / s$.
Let $u$ be a fixed vertex.
Let $A_u$ be the event that there exists a vertex,
not adjacent to $u$, at distance less than $n\sqrt 2 / s$ from $u$.
By the union bound,
$\mathbf{Pr}[A_u] < 2 \pi n^3 (1-p) / s^2$.
The probability that there exists a vertex $u$ for which $A_u$
happens is by the union bound less than
$2 \pi n^4 (1-p) / s^2$.
Therefore,
$$ \int_{s =n^5}^{\infty} \mathbf{Pr} [F > s|\:CON]\: \mathrm{d}s
\leq \int_{s =n^5}^{\infty} \frac{2 \pi n^4 (1-p)}{s^2}\: \mathrm{d}s
 = o(1)\:.$$

We will now bound the second term in the right hand side of (\ref{eq:EF2}).
Notice that for any fixed embedding of the vertices,
the event ``$F > \lambda$'' is a decreasing property (with respect to the edges in the graph),
and the event ``$G$ is connected'' is an increasing one.
Hence by the correlation inequalities (see, e.g., Theorem~6.3.3 in Alon and Spencer~\cite{alon_spencer}) we have
$$\p{F > \lambda |\: CON} \leq \p{F>\lambda}\:.$$

Let $\lambda = (s-1) / 2$. Then by Lemma~\ref{lem:integration}
\begin{align*}
\p{F > s} & \leq n^2 \exp\left(- \frac{p^2n}{16}\right) + \frac{128 n(1-p) }{p^2 \lambda^2} \:.
\end{align*}
Therefore, since $p^2n > 113 \log n$,
\begin{align*}
\int_{s=1+\epsilon}^{n^5} \p{F > s} \:\mathrm{d} s
& \leq \int_{\lambda=\epsilon/2}^{n^5/2} 2n^2 \exp\left(- \frac{p^2n}{16}\right) + \frac{256 n(1-p) }{p^2 \lambda^2} \:\mathrm{d} \lambda\\
& \leq n^7 \exp\left(- \frac{p^2n}{16}\right) + \left. \frac{-256 n(1-p) }{p^2 \lambda} \right|_{\lambda = \epsilon/2}^{n^5/2} \\
& \leq o(1) + \frac{512 n(1-p) }{p^2 \epsilon} = o(1) + \epsilon\:.
\end{align*}
Consequently,
\begin{align*}
\mathbf{E} [F|\: CON] & = \int_{s =0}^{1+\epsilon} \mathbf{Pr} [F > s |\: CON] \:\mathrm{d} s +
\int_{s =1+\epsilon}^{n^5} \mathbf{Pr} [F > s | \:CON] \:ds + \int_{s =n^5}^{\infty} \mathbf{Pr} [F > s |\: CON] \:\mathrm{d} s \\
& \leq 1 + 2\epsilon + o(1) = 1 + {\frac{\sqrt{2048 n(1-p) }}{p}} + o(1)\:.\qedhere
\end{align*}
\end{proof}

%
%

\end{document}